\documentclass[graybox]{svmult}
\usepackage{mathptmx}       
\usepackage{helvet}         
\usepackage{courier}        
\usepackage{makeidx}         
\usepackage{graphicx}        
\usepackage{multicol}        
\usepackage[bottom]{footmisc}
\usepackage{amsmath,amssymb,amsfonts}        
\makeindex             

\begin{document}
\title*{Recent Developments in Heterotic Moduli}
\author{Javier Murgas Ibarra and Eirik Eik Svanes}
\institute{Javier Murgas Ibarra \at University of Stavanger, Kristine Bonnevies vei 22, 4021 Stavanger,\\ \email{javier.j.murgasibarra@uis.no}
\and Eirik Eik Svanes \at University of Stavanger, Kristine Bonnevies vei 22, 4021 Stavanger,\\ \email{eirik.e.svanes@uis.no}}
%
%

\maketitle

\abstract{We review recent results for heterotic moduli and the Hull--Strominger system. In particular, we discuss mathematical properties of the recently derived deformation operator $\bar D$ associated to the deformation complex of heterotic $SU(3)$ solutions. We review results on Serre duality, showing that the operator has a vanishing index, and discuss a notion of \v{C}ech cohomology and a particular instance of a Dolbeault theorem for $\bar D$. Specifically, the cohomology parametrising infinitesimal deformations is isomorphic to the first \v{C}ech cohomology of an associated cochain complex. This will be useful for future research, as it provides a more algebraic handle on the heterotic moduli problem, which is useful for understanding notions of stability, geometric invariants, and enumerative geometry for the Hull--Strominger system.}

\section{Introduction and history}
\label{sec:introduction}
The low-energy effective theory of string theory is ten-dimensional supergravity, which may be used to construct realistic models of grand unified theories in a process known as string compactification. In this scenario, one assumes that six of the ten dimensions are curled up and compact, leaving a four-dimensional spacetime where the effective field theory lives. That is, the spacetime has an ansatz of the form
\begin{equation}
 M_{10} = M_{4} \times X_6^{compact}~.
\end{equation} 
In this setting, the heterotic string is of particular interest, as it is naturally equipped with a gauge bundle wherein the Standard Model fits. The first successful string compactification of this kind was done in 1985 by \cite{Candelas:1985en}, and is known as the Standard Embedding. In this setup, the vector bundle is identified with the tangent bundle. This leads to an exact cancellation of anomalies, which results in the internal geometry of $X_6$ being Kähler and Calabi--Yau. Using Calabi--Yau manifolds as the compact space of string compactifications has turned out to be very effective. Indeed our understanding of Calabi--Yau geometry, their moduli spaces, and constructions through algebraic geometry, has developed significantly over the past few decades, which is also closely tied to their usefulness in string theory.

More general heterotic string compactifications on Calabi--Yau manifolds have also been considered. These are of interest due to their ability to embed chiral gauge theories close to the minimally supersymmetric standard model, see for example \cite{Witten:1986kg, Bouchard:2005ag, Curio:2004pf, Braun:2005ux, Buchmuller:2006ik, Anderson:2011ns, Braun:2013wr, Buchbinder:2014qda, Braun:2017feb, Ashmore:2020ocb}, and their tractability in computing quantum corrections which may be  phrased in terms of geometric conditions on $X_6$. Mathematically, these solutions are known to solve the full set of coupled geometric equations \cite{Andreas:2010qh}, known as the Hull--Strominger system \cite{Hull:1986kz,Strominger:1986uh}. Explicit non-Kähler solutions have also appeared in the literature, such as for example principal $T^2$ fibrations over $K3$ manifolds, see \cite{Dasgupta:1999ss, Becker:2003yv, Becker:2003sh, Goldstein_2004, Fu:2006vj, Becker:2009df, Melnikov:2014ywa, GarciaFernandez2020} and references therein, or local solutions, beginning with the seminal work of \cite{Callan:1991at}, see also \cite{Fu:2009zee, Carlevaro:2009jx, Chen:2013nma, Halmagyi:2016pqu, Halmagyi:2017lqm, Acharya:2020xgn} and references therein. Other mathematically interesting solutions have also appeared, see e.g. \cite{LopesCardoso:2002vpf, Fernandez:2008wa, Fei:2014aca, Fei:2015kua, otal2017invariant, pujia2021hull}. These geometries are far from Kähler Calabi--Yau, but provide a mathematically interesting framework for studying the interplay of gauge theory and complex geometry.

The Hull--Strominger system leads to four-dimensional Minkowski spacetime with $N=1$ supersymmetry. Specifically, the compact geometry $X$ should be a compact complex three-fold with vanishing first Chern class.  There is a holomorphic vector bundle $V$, which admits a connection that is hermitian Yang--Mills. The Green-Schwarz anomaly then relates the hermitian form on $X$ to topological constraints on $V$ and $X$. This somewhat complicated coupling between the gauge and gravitational sector of the Hull--Strominger system makes geometrical aspects such as the moduli problem, or deformation problem, less tractable. Still, there has been recent progress in understanding heterotic moduli, which will partly be reviewed here.

The deformations are coupled via the supersymmetry variations and the heterotic Bianchi identity 
\begin{equation}
    \label{eq:BI1}
    i\partial\bar\partial\omega=\tfrac{\alpha'}{2}\left({\rm tr}(F^2)- {\rm tr}(R^2)\right)\:,
\end{equation}
where $\omega$ is the hermitian two-form, $F$ is the curvature of a connection on a holomorphic gauge bundle with structure group contained in $E_8\times E_8$, and $R$ is the curvature of a connection on the tangent bundle $TX$. We will expand on the details of these structures below. When written in terms of spinor bilinears, and related to the differential geometric structures of the internal manifold $X$, some of these can also be written as the variation of a functional called the superpotential \cite{LopesCardoso:2003dvb,Gurrieri:2004dt,delaOssa:2015maa,McOrist:2016cfl}. We refer to them as F-terms in analogy with $N=1$ $d=4$ supersymmetry. Finite deformations of the F-terms where studied in \cite{Ashmore:2018ybe}, where the corresponding heterotic Maurer-Cartan equation was described.\footnote{This was done ignoring a subtlety regarding a tangent bundle connection which appears in the heterotic Bianchi identity.} The remaining supersymmetry variations are referred to as D-terms, continuing the analogy. Evidence for this was presented in \cite{McOrist:2021dnd}, and further studied in \cite{MPS24}, whereby using the string theory moduli space metric, calculated in \cite{Candelas:2016usb,McOrist:2019mxh} by a dimensional reduction, it was demonstrated that infinitesimal deformations of the F-terms correspond to the kernel of a certain $\bar D$-operator, and the D-terms lie in the kernel of its adjoint ${\bar D}^\dagger$ with respect to the moduli metric. In \cite{Ashmore:2019rkx,Garcia-Fernandez:2020awc} the D-terms were presented in terms of a moment map construction, consistent with the usual $N=1$ $d=4$ supersymmetry lore. 

However, unlike the classical case of the moduli space of Calabi--Yau manifolds (see for example \cite{Candelas:1990pi}), it is certainly not obvious that the deformations decouple into a direct sum of cohomologies, related to complex structure deformations, deformations of the hermitian structure, and deformations of the holomorphic vector bundle. Indeed, as the bundle is holomorphic it requires a field strength $F^{(0,2)} = 0$ with respect to the complex structure of $X$. For a generic deformation of the complex structure, one might generate some $F^{(0,2)}$ that violates this condition. Indeed, as demonstrated by Atiyah \cite{Atiyah:1955}, this is the case if such a deformation is not $\bar{\partial}_{\mathcal A}$--exact. In this case the corresponding complex structure deformation is not a deformation or parameter of the heterotic theory. See \cite{Anderson:2010mh, Anderson:2011ty, Anderson:2011cza, Cicoli:2013rwa} for applications of this mechanism to the heterotic moduli problem. 

The additional coupling of deformations by the Bianchi identity makes the heterotic moduli problem even more involved. There has however been recent progress in decoding the subtleties of this problem. In \cite{Anderson:2014xha, delaOssa:2014cia} the Bianchi identity was encoded as an additional Atiyah type extension, showing that the infinitesimal moduli could be viewed as the first cohomology of a double extension (to a physics level of rigour). Furthermore, in \cite{Garcia-Fernandez:2015hja} the problem was reprhased mathematically as the moduli spaces of solutions of suitable Killing spinor equations on a (holomorphic) Courant algebroid, where it was also shown to be elliptic. A lot of recent progress has since been made in our understanding of the geometric structures of these equations and their moduli, both from a physical and mathematical point of view. A non-exhaustive list of references includes \cite{delaOssa:2014msa, delaOssa:2015maa, Candelas:2016usb, McOrist:2016cfl, Phong:2016ggz, Candelas:2018lib, Garcia-Fernandez:2018emx, Garcia-Fernandez:2018ypt, Garcia-Fernandez:2018qcl, McOrist:2019mxh, Fei:2019dap, Ashmore:2019rkx, Garcia-Fernandez:2020awc, Alvarez-Consul:2020hbl, Smith:2022baw, Collins:2022snx, Garcia-Fernandez:2023vah, Garcia-Fernandez:2023nil, Alvarez-Consul:2023zon, Ashmore:2023ift, Ashmore:2023vji, Tellez-Dominguez:2023wwr, Silva:2024fvl, Kupka:2024rvl, Garcia-Fernandez:2024ypl, Kupka:2024tic}.

Many of the recent developments in heterotic geometry however employ a mathematical trick which is not warranted from a physical point of view. As demonstrated in \cite{Ivanov:2009rh, Martelli:2010jx}, in order for the Hull--Strominger system of equations to solve the equations of motion, it is important that the tangent bundle connection, whose curvature appears in the Bianchi identity \eqref{eq:BI1}, is picked to be an instanton. This is satisfied modulo higher $\alpha'$ corrections for the physical system of equations, where the Hull connection is used \cite{Hull:1986kz}, but unfortunately the resulting physical system of equations is not closed in an appropriate mathematical sense. Thus, in order to get a closed system of equations, it is common to promote this connection to an additional degree of freedom of an instanton on $TX$. This results in a closed system of equations, which is easier to handle mathematically. This, however presents a problem when it comes to studying the moduli problem of the system, as the extra degrees of freedom are non-physical. There are numerous examples historically of physics guiding mathematical insights, and we therefore ignore this subtlety at our peril. The recent developments presented in this review are in the direction of remedying this issue.  

An Alternative approach, as done in Strominger's original paper \cite{Strominger:1986uh}, the connection may be chosen to be the Chern-connection of the hermitian form $\omega$ (which is generically not an instanton, leaving only an approximate solution to the equations of motion). The resulting system of equations may then more accurately be referred to as the Strominger system. See for example \cite{Becker:2006et, Fu:2006vj, Becker:2006xp, Becker:2008rc, Fei:2014aca, Phong:2016gbw, Phong:2016ggz, Fei:2017ctw, Fei:2019dap, Collins:2022snx, Picard:2024tqg, Picard:2024pad} and references therein for this approach. The subtleties regarding the choice of connection are discussed in great detail elsewhere, see for example \cite{Becker:2009df, Melnikov:2014ywa, delaOssa:2014msa, Garcia-Fernandez:2016azr} and references therein. The results presented in the present review are largely applicable to either choosing the Hull connection (modulo higher order $\alpha'$ corrections), the Chern connection, or a hermitian Yang--Mills connection. We will discuss specific instances in how they should be amended for eah choice. 

The paper is organised as follows. In section \ref{sec:HS} we give a brief review of the Hull--Strominger system. Following this, we define the differential and cohomology which counts heterotic moduli in section \ref{sec:Moduli}. We then review some very recent results regarding the mathematical properties of this structure, in particular a Serre duality result in section \ref{sec:Serre}, and a definition of \v{C}ech cohomology and a Dolbeault theorem in section \ref{sec:Cech}. We conclude with some comments on directions of future research in section \ref{sec:Conclusion}.

\section{The Hull--Strominger system}
\label{sec:HS}
We will now summarise Hull--Strominger system of equations for supersymmetry.  The spacetime geometry is taken to be of the form 
\begin{equation}
    M_{10}=\mathbb{R}^{3,1} \times X~.
\end{equation}
The manifold $X$ is six-dimensional and complex, which means that it admits an integrable complex structure 
\begin{equation}
J = J_m{}^n  dx^m \otimes \partial_n\:,    
\end{equation}
where $x^m$ are real coordinates on $X$.  The complex structure is integrable meaning $J^2 = -1$, and its Nijenhuis tensor vanishes. We often have need for wedges -- in the interest of clutter reduction, we  omit the wedge symbol, writing for example $d x^{mn} \cong dx^m \wedge dx^n$. Occasionally, there may be ambiguity in which we case  we write the wedge explicitly.  

The complex manifold has trivial canonical bundle and so admits a holomorphic volume form $\Omega$
\begin{equation}
    \label{eq:HolOm}
    d\Omega =0\:,
\end{equation}
where $\Omega$ is a $(3,0)$-form with respect to the complex structure. We also have a Riemannian metric $d s^2 = g_{mn} d x^m \otimes d x^n$, which is compatible with the complex structure. It follows that there exists a hermitian two-form 
\begin{equation}
\omega = \tfrac12\omega_{mn} d x^{mn} = i g_{\mu\bar\nu} dz^{\mu\bar\nu}\:,    
\end{equation}
where we have evaluated the hermitian form in the complex coordinates that diagonalise a given $J$. 

There is a vector bundle $E \to X$, with a connection $A$, and a structure group being a subgroup of $E_8\times E_8$ or ${\rm Spin}(32)/\mathbb{Z}_2$. The gauge connection is anti-hermitian meaning
$$
A = \mathcal{A} - \mathcal{A}^\dagger~,
$$
where ${\cal A}$ is the $(0,1)$ component of $A$. The field strength for the connection is  $F= d A + A^2$ and supersymmetry requires it be holomorphic 
\begin{equation}
    \label{eq:Ahol}
    F^{(0,2)} = 0\:,
\end{equation}
and  solve a Hermitian Yang--Mills equation
\begin{equation}
    \label{eq:HYM}
    g^{mn} F_{mn} =0~,
\end{equation}
which on $X$ can be alternatively written as $\omega^2 F = 0$.     

There is a three-form $H$ which on the one hand is defined as the contorsion of $\omega$
\begin{equation}
    \label{eq:ddco}
    H = d^c \omega = \frac{1}{3!} J^m J^n J^p (d \omega)_{mnp}~.
\end{equation}
For a fixed complex structure this is 
\begin{equation}
\label{eq:AnomalyTorsion}
H=i (\partial-\bar\partial)\omega\:.    
\end{equation}
The hermitian form is conformally balanced
\begin{equation}
    \label{eq:ConfBal}
    d \left(e^{-2 \Phi} \omega^2\right) =0~,
\end{equation}
where $\Phi = -{1 / 2} \log |\Omega|_\omega$. 

On the other hand the Green--Schwarz anomaly together with coupling of supergravity to a gauge field forces $H$ to satisfy the heterotic Bianchi identity. Given the supersymmetry constraint \eqref{eq:AnomalyTorsion}, this can be written as
\begin{equation}
    \label{eq:BI2}
    i\partial\bar\partial\omega=\tfrac{\alpha'}{2}\left({\rm tr}(F^2)- {\rm tr}(R^2)\right)\:,
\end{equation}
where $R$ is the curvature of a connection on the tangent bundle. As mentioned above, this connection is either chosen to be the Chern connection (which was Strominger's original choice \cite{Strominger:1986uh}), a Hermitian Yang--Mills connection (leading to an exact solution of the heterotic equations of motion), or the Hull connection \cite{Hull:1986kz} (the appropriate choice from a heterotic supergravity point of view, which also solves the equations of motion modulo ${\cal O}(\alpha'^2)$ corrections). We will leave the choice of connection ambiguous, and comment more on the specific choices throughout the text. 

We will refer to equations \eqref{eq:HolOm}, \eqref{eq:Ahol}, \eqref{eq:HYM}, \eqref{eq:ConfBal} and \eqref{eq:BI2} as the Hull--Strominger system, the geometric system of equations under consideration in this text.

\section{Infinitesimal moduli}
\label{sec:Moduli}
It has been known since Atiyah's seminal work \cite{Atiyah:1955} that the simultaneous infinitesimal deformations of the complex structure and holomorphic gauge bundle need not result in the direct sum of cohomologies $H_{\bar\partial}^{(0,1)}(T^{(1,0)}X)$ and $H^{(0,1)}_{\bar{\partial}_{\mathcal A}}(End_0(E))$, counting complex structure deformations and deformations of the holomorphic connection on $E$ respectively. Rather, the corresponding simultaneous deformations are counted by $H^{(0,1)}_{\bar D_1}(Q_1)$, where $\bar D_1$ is defined via a short exact extension
\begin{equation}
    0\rightarrow\Omega_{\bar{\partial}_{\mathcal A}}^{(0,p)}(End_0(E))\rightarrow\Omega^{(0,p)}_{\bar D_1}(Q_1)\rightarrow\Omega_{\bar\partial}^{(0,p)}(T^{(1,0)}X)\rightarrow0\:,
\end{equation}
with extension class $[F]\in H^{(0,1)}(End_0(E)\otimes {T^{*1,0}X})$ given by the curvature of the gauge connection, also called the Atiyah class. This is due to the fact that some directions in the complex structure moduli space need not preserve the holomorphicity of the bundle. Note that the Atiyah class may also be viewed as a map
\begin{equation}
    [F]:\;\;\;H_{\bar\partial}^{(0,1)}(T^{(1,0)}X)\rightarrow H_{\bar{\partial}_{\mathcal A}}^{(0,2)}(End_0(E))\:.
\end{equation}
The infinitesimal deformations may then be computed via the resulting long exact sequence of cohomologies, resulting in\footnote{Here we assume that the gauge bundle is stable, and so there are no holomorphic sections, i.e. $H_{\bar{\partial}_{\mathcal A}}^{0}(End_0(E))=0$.}
\begin{equation}
      H^{(0,1)}_{\bar D_1}(Q_1)\cong H_{\bar{\partial}_{\mathcal A}}^{(0,1)}(End_0(E))\oplus\ker([F])\:,
\end{equation}
where $\ker([F])\subseteq H_{\bar\partial}^{(0,1)}(T^{(1,0)}X)$ are the complex structure moduli which preserve the holomorphicity of the bundle. This mechanism was applied in \cite{Anderson:2010mh, Anderson:2011ty, Anderson:2011cza, Cicoli:2013rwa} in the context of the heterotic moduli problem. 

It was further demonstrated in \cite{Anderson:2014xha, delaOssa:2014cia} that the additional complications of the Green-Schwarz anomaly, or heterotic Bianchi identity, may be encoded as an additional Atiyah type extension. Specifically, the anomaly cancellation gives rise to an Atiyah type map
\begin{equation}
    [{\cal H}]:\;\;\;H_{\bar\partial}^{(0,p)}(Q_1)\rightarrow H_{\bar{\partial}_{\mathcal A}}^{(0,p+1)}(T^{*(1,0)}X)\:,
\end{equation}
which is further used to define the differential $\bar D$ whose first cohomology computes the infinitesimal moduli of the Hull--Strominger system. Specifically, the differential $\bar D$ encodes the heterotic moduli complex as an additonal short exact extension
\begin{equation}
    0\rightarrow\Omega_{\bar{\partial}}^{(0,p)}(T^{*(1,0)}X)\rightarrow\Omega^{(0,p)}_{\bar D}(Q)\rightarrow\Omega_{\bar D_1}^{(0,p)}(Q_1)\rightarrow0\:.
\end{equation}
Again the infinitesimal moduli are computed, using homological algebra, via the corresponding long exact sequence in cohomology as\footnote{We make the physically reasonable assumption that the tangent bundle has no holomorphic sections, $H^{0}(T^{(1,0)}X)=0$.}
\begin{equation}
    H^{(0,1)}_{\bar D}(Q)\cong H^{(0,1)}(T^{*(1,0)}X)\oplus\ker([{\cal H}])\:.
\end{equation}
We may interpret $H^{(0,1)}(T^{*(1,0)}X)$ as the infinitesimal (complexified) hermitian moduli, generalising Kähler moduli in the Calabi--Yau case.

The computations of \cite{Anderson:2014xha, delaOssa:2014cia} however employed the mathematical trick mentioned above, where the connection on the tangent bundle, whose curvature appears in the Bianchi identity, is promoted to its own degree of freedom. This makes the moduli problem more tractable from a mathematical point of view, but means that we are in effect over-counting the physical moduli, with extra ``spurious" modes coming from this connection. This issue was corrected in \cite{McOrist:2021dnd}, where the correct map $[\cal H]$ was identified, giving the correct differential $\bar D$ for the physical heterotic moduli problem, at least modulo ${\cal O}(\alpha'^2)$ corrections. In \cite{MPS24, deLazari:2024zkg}, the operator was corrected to an exact nilpotent operator, whose first cohomology computes the spectrum of the Hull--Strominger system\footnote{Strictly speaking, that $H^{(0,1)}_{\bar D}(Q)$ counts the infinitesimal heterotic moduli is again only known modulo ${\cal O}(\alpha'^2)$ corrections. This is however ok from the point of view of physics, as the Hull--Strominger system gets corrected at ${\cal O}(\alpha'^2)$.}. Schematically, the operator reads
\begin{equation}
\label{eq:Dbar}
    \bar{D} = 
    \begin{bmatrix}
        \bar\partial & {\cal H}  \\
        0 & \bar D_1
    \end{bmatrix}
    =
    \begin{bmatrix}
        \bar\partial\; &\; \alpha' {\cal F}\; &\; {\cal T} + \alpha'  {\cal R}\cdot\nabla^+  \\
        0 & \bar\partial_{\cal A} & {\cal F} \\
        0 & 0 & \bar\partial
    \end{bmatrix}\:,
\end{equation}
where
\begin{equation}
    \bar D:\;\;\Omega^{(0,p)}(Q)\rightarrow\Omega^{(0,p+1)}(Q)\:,
\end{equation}
and
\begin{equation}
    Q=T^{*(1,0)}X\oplus End_0(E)\oplus T^{(1,0)}X\:.
\end{equation}
For $\gamma\in\Omega^{(0,p)}(End_0(V))$ and $\mu\in\Omega^{(0,p)}(T^{(1,0)}X)$ we define 
\begin{equation}
\label{eq:sF}
    \left({\cal F} \gamma\right)_a 
    = tr(F_a\wedge\gamma), \quad {\cal F}\mu = F_b\wedge\mu^b \:,
\end{equation}
where $F_a=F_{a \bar{b}}d\bar{z}^{\bar{b}}$. We also define
\begin{equation}
    \left({\cal T}\mu\right)_a=T_{ba \bar{c}} dz^{\bar c}\wedge\mu^b\:,
\end{equation}
where $T=i\partial\omega$, often referred to as the torsion. 

We further define
\begin{equation}
    \left({\cal R}\cdot\nabla^+\mu\right)_a=R_{a}{}^b{}_c\wedge\nabla^+_b\mu^c\:,
\end{equation}
where $R_a=R_{a \bar{b}}d\bar{z}^{\bar{b}}$ is the curvature of the Chern connection (if one uses the Chern connection in the Bianchi identity). The connection $\nabla^+$ acts as the Chern connection on anti-holomorphic indecies, and as the Bismut connection on holomorphic indecies. That is, 
\begin{equation}
\label{nabla+}
    \nabla^+_aW^b =\nabla_aW^b-T^b{}_{ac}W^c\:,
    \quad
    W\in \Gamma(T^{(1,0)}X)\:,
\end{equation}
where $\nabla$ is the Chern conncetion. Note that if instead one wishes to use a hermitian Yang--Mills connection to define the curvature of the Bianchi identity (solving the heterotic equations of motion exactly), then the curvature of ${\cal R}\cdot\nabla^+$ should be replaced by the curvature of this connection, while $\nabla^+$ should be defined using this hermitian connection, and the torsion of the corresponding Hermitian Yang--Mills metric. 

With these definitions, the operator $\bar D$ is nilpotent, i.e. $\bar D^2=0$, if and only if the heterotic Bianchi identity is satisfied. Furthermore, the infinitesimal moduli of the Hull--Strominger system are counted by $H^{(0,1)}_{\bar D}(Q)$ \cite{McOrist:2021dnd, MPS24, deLazari:2024zkg}. Furthermore, the corresponding moduli complex
\begin{equation}
    \label{eq:Dcomplex}
    0\rightarrow\Omega^{(0,0)}(Q)\xrightarrow{\bar D}\Omega^{(0,1)}(Q)\xrightarrow{\bar D}\Omega^{(0,2)}(Q)\xrightarrow{\bar D}\Omega^{(0,3)}(Q)\rightarrow0
\end{equation}
was shown to be elliptic in \cite{Chisamanga:2024xbm, deLazari:2024zkg}, and so the corresponding cohomologies are finite-dimensional. 

The cohomology $H^{(0,1)}_{\bar D}(Q)$ has been computed in explicit examples, such as for the Standard Embedding \cite{Chisamanga:2024xbm}. Here it was found that the cohomology decomposes into Kähler moduli $H^{(1,1)}(X)$, bundle moduli $H^{(0,1)}(End_0(E))$ and complex structure moduli $H^{(2,1)}(X)$. In particular, they proved the following theorem
\begin{theorem}
    At the Standard embedding \cite{Candelas:1985en}, the cohomology $H^{(0,1)}_{\bar D}(Q)$ decomposes as
    $$H^{(0,1)}_{\bar D}(Q)\cong H^{(1,1)}(X)\oplus H^{(0,1)}(End_0(E))\oplus H^{(2,1)}(X)\:.$$
\end{theorem}
This was previously claimed in the literature, but not demonstrated explicitly until the work of \cite{Chisamanga:2024xbm}. Indeed, the technology to explicitly define and compute the infinitesimal deformations of heterotic solutions was until recently lacking, and remains so to a large extent for more generic solutions. However, $H^{(0,1)}_{\bar D}(Q)$ has also been computed in more non-physical explicit examples that are of mathematical interest, such as for the Iwasawa manifold \cite{deLazari:2024zkg}, as well as the the generalisation of the Hull--Strominger system as a coupled instanton on the Calabi--Eckmann manifold \cite{Silva:2024fvl, deLazari:2024zkg}.

Some additional comments about the operator $\bar D$ are warranted. Note that the presence of the derivative map ${\cal R}\cdot\nabla^+$ means that the operator is not a connection, in that it does not satisfy the Leibniz rule. Treating the connection on the tangent bundle as a separate degree of freedom leads to a similar operator on
\begin{equation}
Q'=T^{*(1,0)}X\oplus End_0(TX)\oplus End_0(E)\oplus T^{(1,0)}X\:,    
\end{equation}
where $End_0(TX)$ correspond to the additional degrees of freedom \cite{Anderson:2014xha, delaOssa:2014cia}. However, in this case the derivative map ${\cal R}\cdot\nabla^+$ is not present in the operator, and so the operator actually defines a holomorphic structure, or the $(0,1)$-part of a connection on $Q'$. This means that mathematical tools such as Serre-duality and \v{C}ech cohomology are readily available. However, this comes at a prise. For one, as noted above, the extra degrees of freedom are non-physical. As a consequence, they come with a negative kinetic terms in the supergravity action. This will inevitably lead to problems when it comes to studying quantum aspects, as fluctuating un-physical degrees of freedom might introduce unwanted anomalies, couplings, etc. Moreover, mathematically, the natural metric on $Q'$ has in-definite signature. This makes it harder to apply Hodge-theory, study flow equations, and prove analogs of the Donalson--Uhlenbeck--Yau theorem \cite{donaldson1985anti, uhlenbeck1986existence}, though progress has been made in this direction, see \cite{Garcia-Fernandez:2023nil, Garcia-Fernandez:2023vah, Garcia-Fernandez:2024ypl}. 

Though not a connection, the operator $\bar D$ on $Q$ mimics many of the features of a holomorphic structure, and many recent developments have been in the direction of understanding to what extent features of ordinary holomorphic structures carry over to $\bar D$. We will review some of these results next.

\section{Serre duality and index}
\label{sec:Serre}
One property enjoyed by a holomorphic vector bundle $V\rightarrow X$ equipped with a holomorphic structure $\bar\partial_{\cal A}$ is Serre duality. In particular, if the complex manifold $X$ is also equipped with a nowhere vanishing holomorphic top-form $\Omega$ (i.e. the canonical bundle is trivial), then
\begin{equation}
    H_{\bar\partial_{\cal A}}^{(0,p)}(V)\cong H^{(0,n-p)}(V^*)\:,
\end{equation}
where $n={\rm dim}_{\mathbb{C}}(X)$. As the topological sum of bundles making up $Q$ is self-dual, one might expect this to be true as well for the cohomologies of $\bar D$, in the sense that $Q$ is ``cohomologically self-dual", 
\begin{equation}
    H^{(0,p)}_{\bar D}(Q)\cong H^{(0,3-p)}_{\bar D}(Q)\:.
\end{equation}
Using ordinary Serre-duality techniques for holomorphic structures, this can be shown in the case when the spurious degrees of freedom are included \cite{Chisamanga:2024xbm, deLazari:2024zkg}. In fact, in \cite{deLazari:2024zkg} it was also shown to be true for the $\bar D$ operator \eqref{eq:Dbar}, where the following theorem was proven
\begin{theorem}
    The cohomology groups $H^{(0,p)}_{\bar D}(Q)$ satisfy the Serre-duality property 
    $$H^{(0,p)}_{\bar D}(Q)\cong H^{(0,3-p)}_{\bar D}(Q)\:.$$
    \end{theorem}

The fact that the operator $\bar D$ enjoys a version of Serre duality has some interesting consequences. In particular, it means that 
\begin{equation}
    H^{(0,1)}_{\bar D}(Q)\cong H^{(0,2)}_{\bar D}(Q)\:.
\end{equation}
As $H^{(0,1)}_{\bar D}(Q)$ counts the infinitesimal deformations of the Hull--Strominger system, while $H^{(0,2)}_{\bar D}(Q)$ parametrises the space of obstructions, this means that the ``mathematically expected", or ``virtual", moduli space of the Hull--Strominger system is zero-dimensional! I.e. the expected moduli space is a set of points. With this in mind, one may speculate as to whether one can ``count'' the number of solutions to the Hull--Strominger system on a compact 6-manifold $X$, and if this carries any interesting information about $X$ or the chosen bundles over it.  Defining such a count motivates the study of the compactness theory for solutions. If such a well-defined number can be assigned to the solution space, the next intriguing question would be if, and in what way, it is invariant under suitable perturbations. 

The moduli space having an expected dimension of zero is also interesting from a physical point of view. Indeed, it is well known to be difficult in string theory to find (perturbative) string compactifications where all the moduli are ``stabilised''. This translates mathematically to having a zero-dimensional moduli space, and is known as the string theory moduli problem. One might hope that the moduli are stabilised by some unknown higher order perhaps non-perturbative or quantum effect, and the fact that the moduli space has expected dimension zero lends credence to this idea. See also \cite{Moroianu:2021kit, Moroianu:2023jof}, where rigidity results have been found for more general non-supersymmetric solutions to heterotic supergravity.

As a corollary to the Serre-duality result, it is straightforward to see that the heterotic deformation complex governed by $\bar D$ has vanishing Euler characteristic, that is
\begin{equation}
\chi(\bar D,Q)=\sum_{p=0}^3(-1)^p h^p(Q)=0\:,    
\end{equation}
where $h^p(Q)$ are the dimensions of the corresponding cohomology groups $H^{(0,p)}_{\bar D}(Q)$. Working within the $\alpha'$ expansion, it was shown in \cite{McOrist:2021dnd} that the infinitesimal moduli parametrising deformations of the Hull--Strominger system must also be in the kernel of an adjoint operator $\bar D^\dagger$, where the adjoint is defined with respect to the moduli space metric \cite{Candelas:2016usb, Candelas:2018lib, McOrist:2019mxh}, see also \cite{Ashmore:2019rkx, Garcia-Fernandez:2020awc}. For $\alpha'$ sufficiently small this metric is positive definite. One may then use Hodge theory to conclude that $H^{(0,1)}_{\bar D}(Q)$ count the infinitesimal deformations of the Hull--Strominger system. Furthermore, one can define a Dirac type operator 
\begin{equation}
    {\cal D}=\bar D+\bar D^\dagger\:,
\end{equation}
whose index agrees with $\chi(\bar D,Q)$, that is
\begin{equation}
    {\rm Index}({\cal D})=\chi(\bar D,Q)=0\:.
\end{equation}
This shows that the moduli problem of the Hull--Strominger system is an ``index zero moduli problem''. Such moduli problems have a tendency to be more mathematically well-behaved, both from a differential geometry and algebraic geometry point of view.

The vanishing of the index, or Euler Characteristic, is a straightforward consequence of the Serre duality result. There is, however, a simpler proof of the index vanishing, which is more in line with the topological nature of the index, which we present here. 

\begin{theorem}
    The complex \eqref{eq:Dcomplex} has vanishing Euler characteristic.
\end{theorem}
\begin{proof}
    The proof relies on constructing a path of elliptic complexes for operators $\bar D_t$, for $t\in[0,1]$, where $\bar D_1=\bar D$ and $\bar D_0$ is the diagonal operator on $Q=T^{*(1,0)}X\oplus End_0(E)\oplus T^{(1,0)}X$. Specifically
        \begin{equation}
            \bar{D}_t =
        \begin{bmatrix}
            \bar\partial\; &\; t(\alpha' {\cal F})\; &\; t^2\left({\cal T} + \alpha'  {\cal R}\cdot\nabla^+\right)  \\
            0 & \bar\partial_{\cal A} & t({\cal F}) \\
            0 & 0 & \bar\partial
        \end{bmatrix}\:.
        \end{equation}
    It is easy to check that the $\bar D_t$ are nilpotent, as $\bar D_0$ is clearly nilpotent, and for $t\neq0$ we have $\bar D_t^2=0$ if and only if $\bar D^2=0$. Furthermore, the corresponding complexes $\Omega_{\bar D_t}^*(Q)$ are elliptic by a similar argument found in \cite{Chisamanga:2024xbm} for the ellipticity of the complex $\Omega_{\bar D}^*(Q)$ in \eqref{eq:Dcomplex}. 
    
    We thus have a path of elliptic complexes from $\Omega_{\bar D}^*(Q)$ to $\Omega_{\bar D_0}^*(Q)$ with the diagonal holomorphic structure. The Euler characteristic is preserved under such a path. Furthermore, as the bundle $Q$ with the diagonal holomorphic structure is clearly self-dual, its Euler characteristic vanishes by ordinary Serre duality. The result follows. \qed 
\end{proof}

\section{\v{C}ech cohomology and a Dolbeault theorem}
\label{sec:Cech}
In this section we review how a solution to the Hull--Strominger system defines a notion of \v{C}ech cohomology. We also present (part of) a Dolbeault-type theorem, namely that the first \v{C}ech cohomology group is isomorphic to the first cohomology group defined by the deformation operator $\bar D$, that is $H^{(0,1)}_{\bar D}(Q)$ which parametrises infinitesimal moduli as reviewed above. This more algebraic perspective on the cohomology should prove useful in understanding and computing the deformation and obstruction spaces for heterotic $SU(3)$ solutions.

Let's recall that a Dolbeault operator $\bar{\partial}_V$,  or holomorphic structure, on a rank $r$ complex bundle $\pi :
V \to X$ is an operator
\begin{equation}
    \bar{\partial}_V: \Gamma(\Lambda^{p} T^{*(0,1)}X \otimes E  ) \to \Gamma(\Lambda^{p+1} T^{*(0,1)}X \otimes E  )
\end{equation}
satisfyng the Cauchy equation ($\bar{\partial}_V^2=0$) and the graded Leibniz rule. We would like to view $\bar{D}$ as a “Dolbeault-type” operator with the
nilpotency $\bar{D}^2=0$ as an integrability condition. However, such an operator does not satisfy Leibniz rule, and as such does not define a holomorphic structure. It would then interesting to see what of the various classical results about holomorphic vector bundles still hold or how they are affected. This might be important for future works in order to study notions of stability and get a better understanding of the moduli space of the Hull--Strominger system, and potentially geometric invariants and enumarive geometry for heterotic geometries. 

A basic result of Dolbeault operators is that they can be identified locally with the usual Dolbeault operator $\bar\partial$. For suitable open sets $U_i$ of
$X$ and trivialisations $\phi_i: \pi_i^{-1}(U_i) \to U_i \times \mathbb{C}^r $ one has
\begin{equation}
     \bar{\partial}_V=\phi_{i}^{-1} \circ \bar \partial \circ \phi_i
\end{equation}
Similarly, one can show that the operator $\bar{D}$ also admits trivializations on contractible patches, i.e,
\begin{equation}
\label{eq:LocTriv}
    \bar{D}=\phi_{i}^{-1} \circ \bar \partial \circ \phi_i
\end{equation}
on some contractible open sets $U_i$ \cite{deLazari:2024zkg}. Schematically, the maps $\phi_i$ and $\phi_i^{-1}$ take the form
\begin{align}
    \phi_i &=
        \begin{bmatrix}
            1\; &\; \alpha' A_i\; &\; {\tau}_i + \alpha'  {\Gamma}_i\cdot\nabla^+  \\
            0 & 1 & A_i \\
            0 & 0 & 1
        \end{bmatrix}\:,\\
    \notag\\
    \phi_i^{-1}&=
        \begin{bmatrix}
            1\; &\; -\alpha' A_i\; &\; \alpha' A_i\cdot A_i -({\tau}_i + \alpha'  {\Gamma}_i\cdot\nabla^+)  \\
            0 & 1 & -A_i \\
            0 & 0 & 1
        \end{bmatrix}\:.
\end{align}
Here $A_i$ is the local connection one-form on the gauge bundle, and $\Gamma_i$ is the local one-form of the Chern connection of either $\omega$ (using the Chern connection in the Bianchi identity), or the Hermitian Yang--Mills metric (using a Hermitian Yang--Mills connection in the Bianchi identity). Furthermore, $\tau_i\in\Gamma(U_i,T^{*(1,0)}X\otimes T^{*(1,0)}X)$ is defined locally by a decent procedure using the heterotic Bianchi identity. It should be noted that $\tau_i$ need not be anti-symmetric in its indices, as is the case for the $B$-field in more conventional generalised geometry.  

We note that the complication here is that, due to the term $\alpha'  {\cal R}\cdot\nabla^+$ in equation \eqref{eq:Dbar}, the trivializations  $\phi_i$ involve holomorphic derivatives that prevent $\phi_i$ of being $C^{\infty}(X)$-linear, and therefore not an element of $ \Gamma(U_i, End_0(Q))$, but rather a linear homomorphism. 

Another classical result of complex geometry that we consider here, is the Dolbeault theorem, which relates the Dolbeault cohomology of a holomorphic vector bundle $V \to X$, with the \v{C}ech cohomology of the bundle in the following way.  
\begin{equation}
    H_{\bar \partial}^{(p,q)}(V) \cong \check{H}^{q}(\Omega ^p \otimes V)  
\end{equation}
where $\Omega^p$ is the bundle of holomorphic $p$-forms. It is interesting to know whether there is a Dolbeault type theorem for this operator. In \cite{deLazari:2024zkg}, and analogous result was proven for a particular case. We use the trivializations before to define \v{C}ech Cohomology. Let $\mathcal{U}=\{ U_i \}_{i \in I}$ be a good cover of $X$ and denote $U_{i_1,i_2,...,i_k}:= \bigcap_{l=1}^{k} U_{i_l}$. On $U_{i,j}$ we define the transition maps
\begin{equation}
    \psi_{ij}:= \phi_i \phi_j^{-1} : \Gamma(Q \restriction (U_{i,j}))\to \Gamma(Q \restriction (U_{i,j}))\:. 
\end{equation}
It is straight forward to check that the transition maps are holomorphic, in the sense that $[\bar\partial,\psi_{ij}]=0$. They also satisfy the cocycle conditions 
\begin{equation}
    \psi_{ij}\circ\psi_{ji}={\rm id}\;\;{\rm on}\;\; U_{i,j}\;\;\;{\rm and}\;\;\;\psi_{ij}\circ\psi_{jk}\circ\psi_{ki}={\rm id}\;\;{\rm on}\;\; U_{i,j,k}\:,
\end{equation}
as for transition functions of ordinary holomorphic bundles. 

The bundle $Q$ together with the transition maps $\psi_{ij}$, which we may denote $(Q,\psi_{ij})$, then satisfy many properties of holomorphic bundles. For example, we can define a notion of \v{C}ech cohomology which we review next. The cochain complex is defined by:
\begin{equation}
    \check{C}^k(\Omega^0(Q))=\prod\{ 
 \Omega^0(U_{i_0,...,i_k},Q) : i_0<i_1<...<i_k   \}\:,
\end{equation} 
where $\Omega^0(U_{i_0,...,i_k},Q)$ denotes the set of $\bar\partial$-holomorphic $Q$-valued functions on $U_{i_0,...,i_k}$. The coboundary operator $\delta$ is then defined on a \v{C}ech $p$-cochain $\{\gamma_{i_0,...,i_k}\}\in\check{C}^k(\Omega^0(Q))$ by 
\begin{equation}
(\delta\gamma)_{i_0,...,i_{k+1}}=\psi_{i_0 i_1} \gamma_{i_1,...,i_{k+1}} + \sum_{l=1}^{k+1} (-1)^l \gamma_{i_0,...,\hat{i}_l,...,i_{k+1}}  
\end{equation}
One can check that $\delta^2=0$, so it defines a \v{C}ech cohomology $\check{H}^{k} (\Omega^0(Q))$. Also, $[\delta, \bar\partial]=0$, which gives it good homological properties. With this definition of \v{C}ech cohomology, the following Dolbeault type theorem was proven in \cite{deLazari:2024zkg}:
\begin{theorem}
\label{Tm:ChechIso}
For $\bar D$-cohomology, and \v{C}ech cohomology we have the following
$$H_{\bar D}^{(0,1)}(Q) \cong \check{H}^{1} (\Omega^0(Q))\:.$$
\end{theorem}
The proof of the theorem was done by constructing an explicit isomorphism between the cohomology groups. We will not repeat the proof here, but we will show how a class $[y]\in H^{(0,1)}_{\bar D}(Q)$ gives rise to a $\delta$-closed cochain $\{\gamma_{i,j}\}\in\check{C}^1(\Omega^0(Q))$.

Start by noting that for a representative $y\in[y]$, we have
\begin{equation}
    \bar D y=0\:.
\end{equation}
By the local trivialisation \eqref{eq:LocTriv}, is is straight-forward to show that $\bar D$ satisfies a Poincaré lemma. Therefore, locally on a patch $U_i$, we have
\begin{equation}
    y=\bar D\eta_i\:,
\end{equation}
for $\eta_i\in\Omega^{(0,0)}(U_i,Q)$. On an intersection $U_{i,j}$ (for $i<j$), as $y$ is globally defined, we must have
\begin{equation}
    \bar D(\eta_i-\eta_j)=0\;\;\Rightarrow\;\;\phi_i^{-1}\circ\bar\partial\circ\phi_i(\eta_i)-\phi_j^{-1}\circ\bar\partial\circ\phi_j(\eta_j)=0\:.
\end{equation}
It follows that 
\begin{equation}
    \bar\partial\circ\phi_i(\eta_i)-\psi_{ij}\circ\bar\partial\circ\phi_j(\eta_j)=0\:.
\end{equation}
Furthermore, since $[\bar\partial,\psi_{ij}]=0$, we get
\begin{equation}
    \bar\partial\left(\phi_i(\eta_i)-\phi_i(\eta_j)\right)=0\:.
\end{equation}
We define 
\begin{equation}
    \gamma_{i,j}=\phi_i(\eta_i)-\phi_i(\eta_j)\:.
\end{equation}
This is then an element $\{\gamma_{i,j}\}$ of $\check{C}^1(\Omega^0(Q))$. Furthermore, it is straight forward to show that $\delta\gamma=0$, and so $\{\gamma_{i,j}\}$ defines a class in $\check{H}^{1} (\Omega^0(Q))$. This defines a map from $H_{\bar D}^{(0,1)}(Q)$ to $\check{H}^{1} (\Omega^0(Q))$. One then needs to check that this map is well-defined, and has an inverse. This is done in \cite{deLazari:2024zkg}. It should be noted that the constructions of these maps follows a common recipe in \v{C}ech cohomology, using a partition of unity, except that this recipe must be corrected to account for the derivative terms in the maps $\psi_{ij}$.

It would be very interesting to extend Theorem \ref{Tm:ChechIso} to a general Dolbeault type theorem for cohomologies in general dimension, perhaps also for heterotic type geometries in dimensions other than $6$. In particular, as \v{C}ech cohomology is more algebraic in nature, this would help towards understanding the Hull--Strominger system from the perspective of algebraic geometry, and perhaps shed light on notions of stability for this system, and heterotic enumerative geometry. 

\section{Conclusions and future directions}
\label{sec:Conclusion}
This article reviews recent developments in heterotic moduli, for six-dimensional heterotic $SU(3)$ solutions (the Hull--Strominger system), mostly presented in the papers \cite{McOrist:2021dnd, MPS24, Chisamanga:2024xbm, deLazari:2024zkg}. In particular, we discuss the nilpotent differential $\bar D$ associated to the heterotic moduli complex, whose first cohomology parametrises infinitesimal deformations of such solutions. This operator has many features which mimic a holomorphic structure, but is not a connection due to its failure of satisfying the Leibniz rule. Still, many interesting results can be derived, such as results about Serre-duality and the vanishing index of the heterotic moduli problem, and a definition of \v{C}ech cohomology and Dolbeault type theorems. 

There are many interesting future directions. It would be interesting to compute $H^{(0,1)}_{\bar D}(Q)$ for more examples, and to get a better handle on the ``Atiyah map" ${\cal H}$ in more generic situations. A good place to begin is perhaps backgrounds where the geometry admits a (zeroth order) Calabi--Yau metric. It would also be interesting to see if bounds can be found on the true number of (unobstructed) moduli, see e.g. \cite{Anderson:2011ty, delaOssa:2015maa} for previous work on this, and also \cite{Smith:2024ejf} in the case of type II solutions. It is however likely that this requires a better understanding of obstructions and the corresponding Maurer-Cartan equation (see below).

As also mentioned, it would be interesting to see if the Dolbeault theorem reviewed above can be generalised to higher cohomology groups, and to other dimensions and geometric structures where similar transition maps involving derivatives appear. One can ask questions about the construction of a \v{C}ech cohomology Atiyah class and its relation to a notion of curvature, and classifications in terms of differential cohomology. In particular, if a well-defined notion of curvature can be defined, one can then try to relate this to a Hermitian Yang--Mills constraint for heterotic geometries. In \cite{MPS24} a notion of curvature for $\bar D$ was defined modulo ${\cal O}(\alpha'^2)$ corrections. It was then shown that the Hull--Strominger system is satisfied if and only if this curvature is ``Hermitian Yang--Mills", again modulo ${\cal O}(\alpha'^2)$ corrections. It would be interesting to see if this statement can be made mathematically more precise, and possibly connect it to a statement about stability for heterotic solutions ala Donalson--Uhlenbeck--Yau. This would complement previous work in this direction \cite{Garcia-Fernandez:2023vah, Garcia-Fernandez:2023nil, Silva:2024fvl, Garcia-Fernandez:2024ypl}, where the spurious degrees of freedom where turned on. 

Another interesting question relates to the observation that
\begin{equation}
    H^{(0,1)}_{\bar D}(Q)\cong H^{(0,2)}_{\bar D}(Q)\:,   
\end{equation}
suggesting that the moduli space has expected, or ``virtual", dimension zero, perhaps in some ``derived stacky sense". In particular, can one define a perfect obstruction theory, and a notion of virtual fundamental class on the (derived) moduli space? This would be very useful for understanding geometric invariants and enumerative geometry of heterotic geometries, and the \v{C}ech cohomology reviewed above is expected to be useful in this regard. There is also work in this direction when one includes the non-physical spurious modes, see work by Tellez-Dominguez \cite{Tellez-Dominguez:2023wwr} explaining the right notion of higher holomorphic data to consider, and upcoming papers with co-authors. It would be very interesting to see how the ``holomorphic" operator $\bar D$ and its associated structures fit into this picture, and into the framework of higher geometry and gauge theory in general.

From a more pragmatic point of view, if one wants to have a better understanding and handle on obstructions, it is important to consider higher order deformations of the system. In \cite{Ashmore:2018ybe} the full Maurer--Cartan equation for the Hull--Strominger system was derived, though again with spurious modes turned on. It was shown that the natural $L_\infty$ algebra to consider for the deformation theory was an $L_3$ algebra, with close connections to the holomorphic Courant algebroid. Physically, it is tempting to speculate that the finite deformation problem, without spurious modes, will retain the natural structure of an $L_3$ algebra. This is because the Maurer--Cartan equation of \cite{Ashmore:2018ybe} can be seen as the equation of motion of a superpotential functional, derived for heterotic geometries in previous work \cite{Becker:2003yv, LopesCardoso:2003dvb, Gurrieri:2004dt}. Physics dictates that couplings between fields in the superpotential of quartic order and higher correspond to irrelevant operators, which to a large extent can be ignored. It would be interesting to see if this expectation holds true.

Finally, note that the superpotential theory of \cite{Ashmore:2018ybe} has many features similar to that of Kodaira--Spencer gravity \cite{Bershadsky:1993cx} and Donaldson--Thomas theory \cite{donaldson1998gauge, thomas1997gauge}. The holomorphic Chern--Simons action then give Donaldson--Thomas invariants a natural physical interpretation, in that they can be related to the partition function and other correlation functions of this theory. The heterotic superpotential provides a natural generalisation of this theory, with both geometric and gauge degrees of freedom turned on. Computing its partition function and other correlators would then give a more hands on physical approach to understanding geometric invariants in heterotic geometry. Some preliminary work in this direction may be found in \cite{Ashmore:2023vji, Ibarra:2024hfm}. In particular, in \cite{Ashmore:2023vji} the one-loop partition function of the quadratic action was computed and studied, and its dependence on the background geometry was investigated to determine to what degree it defines an invariant. It would be interesting to generalise such computations to the full superpotential action, without spurious modes, and to investigate the potential anomalies and geometric dependence of the partition function and other quantum correlators. In particular, this might lead to heterotic version of the holomorphic anomaly equation \cite{Bershadsky:1993ta}, a useful tool to bootstrap higher loop order (or higher genus) invariants from lower order ones. 

As a final remark, it should be mentioned that many of the above questions may also be posed for other geometric solutions to the heterotic string, such as exceptional $G_2$ and $Spin(7)$ compactifications. Indeed, these solutions provide a natural mathematical framework for understanding gauge theory coupled to geometry, in that they give rise to ``coupled instantons" in the sense of \cite{Silva:2024fvl}. In particular, the moduli problem of heterotic $G_2$ solutions was previously studied in \cite{Clarke:2016qtg, delaOssa:2017pqy, deIaOssa:2019cci, Clarke:2020erl, Kupka:2024rvl}, but again with the caveat that they include spurious modes of an associated tangent bundle connection. Recently, steps to remedy this issue was taken in \cite{McOrist:2024ivz}, providing a foothold to address such questions, without the spurious modes, in these more exotic settings.

\begin{acknowledgement}
We would like to thank Luis Álvarez-Cónsul, Anthony Ashmore, Beatrice Chisamanga, Xenia de la Ossa, Hannah de Lázari, Jason D. Lotay, Mario Garcia-Fernandez, Julian Kupka, Ruben Minasian, Raul Gonzalez Molina, David Duncan McNutt, Jock McOrist, Ilarion V. Melnikov, Sébastien Picard, Henrique Sá Earp, Christian Saemann, Savdeep Sethi, George R. Smith, Charles Strickland-Constable, Roberto Tellez-Dominguez, David Tennyson, Markus Upmeier, Fridrich Valach, Daniel Waldram, and Sander Winje for enlightening conversations surrounding the research presented here. ES would like to thank the mathematical research institute MATRIX in Australia, where part of the research reviewed here was performed, for a lively and rewarding research environment.
\end{acknowledgement}

\bigskip


\begin{thebibliography}{100}
\providecommand{\url}[1]{{#1}}
\providecommand{\urlprefix}{URL }
\expandafter\ifx\csname urlstyle\endcsname\relax
  \providecommand{\doi}[1]{DOI~\discretionary{}{}{}#1}\else
  \providecommand{\doi}{DOI~\discretionary{}{}{}\begingroup
  \urlstyle{rm}\Url}\fi

\bibitem{Acharya:2020xgn}
Acharya, B.S., Kinsella, A., Svanes, E.E.: {$T^{3}$-invariant heterotic
  Hull-Strominger solutions}.
\newblock JHEP \textbf{01}, 197 (2021).
\newblock \doi{10.1007/JHEP01(2021)197}

\bibitem{Alvarez-Consul:2023zon}
\'Alvarez-C\'onsul, L., de~La~Hera, A.D.A., Garcia-Fernandez, M.: {Vertex
  algebras from the Hull-Strominger system}  (2023)

\bibitem{Alvarez-Consul:2020hbl}
\'Alvarez-C\'onsul, L., de~La~Hera, A.D.A., Garcia-Fernandez, M.: {(0,2) Mirror
  Symmetry on Homogeneous Hopf Surfaces}.
\newblock Int. Math. Res. Not. \textbf{2024}(2), 1211--1298 (2024).
\newblock \doi{10.1093/imrn/rnad016}

\bibitem{Anderson:2011cza}
Anderson, L.B., Gray, J., Lukas, A., Ovrut, B.: {Stabilizing All Geometric
  Moduli in Heterotic Calabi-Yau Vacua}.
\newblock Phys. Rev. \textbf{D83}, 106,011 (2011).
\newblock \doi{10.1103/PhysRevD.83.106011}

\bibitem{Anderson:2010mh}
Anderson, L.B., Gray, J., Lukas, A., Ovrut, B.: {Stabilizing the Complex
  Structure in Heterotic Calabi-Yau Vacua}.
\newblock JHEP \textbf{02}, 088 (2011).
\newblock \doi{10.1007/JHEP02(2011)088}

\bibitem{Anderson:2011ty}
Anderson, L.B., Gray, J., Lukas, A., Ovrut, B.: {The Atiyah Class and Complex
  Structure Stabilization in Heterotic Calabi-Yau Compactifications}.
\newblock JHEP \textbf{10}, 032 (2011).
\newblock \doi{10.1007/JHEP10(2011)032}

\bibitem{Anderson:2011ns}
Anderson, L.B., Gray, J., Lukas, A., Palti, E.: {Two Hundred Heterotic Standard
  Models on Smooth Calabi-Yau Threefolds}.
\newblock Phys. Rev. D \textbf{84}, 106,005 (2011).
\newblock \doi{10.1103/PhysRevD.84.106005}

\bibitem{Anderson:2014xha}
Anderson, L.B., Gray, J., Sharpe, E.: {Algebroids, Heterotic Moduli Spaces and
  the Strominger System}.
\newblock JHEP \textbf{07}, 037 (2014).
\newblock \doi{10.1007/JHEP07(2014)037}

\bibitem{Andreas:2010qh}
Andreas, B., Garcia-Fernandez, M.: {Solutions of the Strominger System via
  Stable Bundles on Calabi-Yau Threefolds}.
\newblock Commun. Math. Phys. \textbf{315}, 153--168 (2012).
\newblock \doi{10.1007/s00220-012-1509-9}

\bibitem{Ashmore:2018ybe}
Ashmore, A., De~La~Ossa, X., Minasian, R., Strickland-Constable, C., Svanes,
  E.E.: {Finite deformations from a heterotic superpotential: holomorphic
  Chern--Simons and an $L_\infty$ algebra}.
\newblock JHEP \textbf{10}, 179 (2018).
\newblock \doi{10.1007/JHEP10(2018)179}

\bibitem{Ashmore:2020ocb}
Ashmore, A., Dumitru, S., Ovrut, B.A.: {Line Bundle Hidden Sectors for Strongly
  Coupled Heterotic Standard Models}.
\newblock Fortsch. Phys. \textbf{69}(7), 2100,052 (2021).
\newblock \doi{10.1002/prop.202100052}

\bibitem{Ashmore:2023vji}
Ashmore, A., Ibarra, J.J.M., McNutt, D.D., Strickland-Constable, C., Svanes,
  E.E., Tennyson, D., Winje, S.: {A heterotic Kodaira-Spencer theory at
  one-loop}.
\newblock JHEP \textbf{10}, 130 (2023).
\newblock \doi{10.1007/JHEP10(2023)130}

\bibitem{Ashmore:2023ift}
Ashmore, A., Minasian, R., Proto, Y.: {Geometric Flows and Supersymmetry}.
\newblock Commun. Math. Phys. \textbf{405}(1), 16 (2024).
\newblock \doi{10.1007/s00220-023-04910-7}

\bibitem{Ashmore:2019rkx}
Ashmore, A., Strickland-Constable, C., Tennyson, D., Waldram, D.: {Heterotic
  backgrounds via generalised geometry: moment maps and moduli}.
\newblock JHEP \textbf{11}, 071 (2020).
\newblock \doi{10.1007/JHEP11(2020)071}

\bibitem{Atiyah:1955}
Atiyah, M.F.: Complex analytic connections in fibre bundles.
\newblock Trans. AMS \textbf{85}(1), 181--207 (1957).
\newblock \doi{10.2307/1992969}

\bibitem{Becker:2003yv}
Becker, K., Becker, M., Dasgupta, K., Green, P.S.: {Compactifications of
  heterotic theory on nonKahler complex manifolds. 1.}
\newblock JHEP \textbf{04}, 007 (2003).
\newblock \doi{10.1088/1126-6708/2003/04/007}

\bibitem{Becker:2006et}
Becker, K., Becker, M., Fu, J.X., Tseng, L.S., Yau, S.T.: {Anomaly cancellation
  and smooth non-Kahler solutions in heterotic string theory}.
\newblock Nucl. Phys. B \textbf{751}, 108--128 (2006).
\newblock \doi{10.1016/j.nuclphysb.2006.05.034}

\bibitem{Becker:2003sh}
Becker, K., Becker, M., Green, P.S., Dasgupta, K., Sharpe, E.:
  {Compactifications of heterotic strings on nonKahler complex manifolds. 2.}
\newblock Nucl. Phys. B \textbf{678}, 19--100 (2004).
\newblock \doi{10.1016/j.nuclphysb.2003.11.029}

\bibitem{Becker:2009df}
Becker, K., Sethi, S.: {Torsional Heterotic Geometries}.
\newblock Nucl. Phys. \textbf{B820}, 1--31 (2009).
\newblock \doi{10.1016/j.nuclphysb.2009.05.002}

\bibitem{Becker:2006xp}
Becker, M., Tseng, L.S., Yau, S.T.: {Moduli Space of Torsional Manifolds}.
\newblock Nucl.Phys. \textbf{B786}, 119--134 (2007).
\newblock \doi{10.1016/j.nuclphysb.2007.07.006}

\bibitem{Becker:2008rc}
Becker, M., Tseng, L.S., Yau, S.T.: {New Heterotic Non-Kahler Geometries}.
\newblock Adv. Theor. Math. Phys. \textbf{13}(6), 1815--1845 (2009).
\newblock \doi{10.4310/ATMP.2009.v13.n6.a5}

\bibitem{Bershadsky:1993ta}
Bershadsky, M., Cecotti, S., Ooguri, H., Vafa, C.: {Holomorphic anomalies in
  topological field theories}.
\newblock Nucl. Phys. B \textbf{405}, 279--304 (1993).
\newblock \doi{10.1016/0550-3213(93)90548-4}

\bibitem{Bershadsky:1993cx}
Bershadsky, M., Cecotti, S., Ooguri, H., Vafa, C.: {Kodaira-Spencer theory of
  gravity and exact results for quantum string amplitudes}.
\newblock Commun. Math. Phys. \textbf{165}, 311--428 (1994).
\newblock \doi{10.1007/BF02099774}

\bibitem{Bouchard:2005ag}
Bouchard, V., Donagi, R.: {An SU(5) heterotic standard model}.
\newblock Phys.Lett. \textbf{B633}, 783--791 (2006).
\newblock \doi{10.1016/j.physletb.2005.12.042}

\bibitem{Braun:2017feb}
Braun, A.P., Brodie, C.R., Lukas, A.: {Heterotic Line Bundle Models on
  Elliptically Fibered Calabi-Yau Three-folds}.
\newblock JHEP \textbf{04}, 087 (2018).
\newblock \doi{10.1007/JHEP04(2018)087}

\bibitem{Braun:2013wr}
Braun, V., He, Y.H., Ovrut, B.A.: {Supersymmetric Hidden Sectors for Heterotic
  Standard Models}.
\newblock JHEP \textbf{09}, 008 (2013).
\newblock \doi{10.1007/JHEP09(2013)008}

\bibitem{Braun:2005ux}
Braun, V., He, Y.H., Ovrut, B.A., Pantev, T.: {A Heterotic standard model}.
\newblock Phys.Lett. \textbf{B618}, 252--258 (2005).
\newblock \doi{10.1016/j.physletb.2005.05.007}

\bibitem{Buchbinder:2014qda}
Buchbinder, E.I., Constantin, A., Lukas, A.: {A heterotic standard model with
  $B - L$ symmetry and a stable proton}.
\newblock JHEP \textbf{06}, 100 (2014).
\newblock \doi{10.1007/JHEP06(2014)100}

\bibitem{Buchmuller:2006ik}
Buchmuller, W., Hamaguchi, K., Lebedev, O., Ratz, M.: {Supersymmetric Standard
  Model from the Heterotic String (II)}.
\newblock Nucl. Phys. B \textbf{785}, 149--209 (2007).
\newblock \doi{10.1016/j.nuclphysb.2007.06.028}

\bibitem{Callan:1991at}
Callan Jr., C.G., Harvey, J.A., Strominger, A.: {Supersymmetric string
  solitons}  (1991)

\bibitem{Candelas:2018lib}
Candelas, P., De~La~Ossa, X., McOrist, J., Sisca, R.: {The Universal Geometry
  of Heterotic Vacua}.
\newblock JHEP \textbf{02}, 038 (2019).
\newblock \doi{10.1007/JHEP02(2019)038}

\bibitem{Candelas:1985en}
Candelas, P., Horowitz, G.T., Strominger, A., Witten, E.: {Vacuum
  Configurations for Superstrings}.
\newblock Nucl. Phys. \textbf{B258}, 46--74 (1985).
\newblock \doi{10.1016/0550-3213(85)90602-9}

\bibitem{Candelas:1990pi}
Candelas, P., de~la Ossa, X.: {Moduli space of Calabi-Yau manifolds}.
\newblock Nucl. Phys. \textbf{B355}, 455--481 (1991).
\newblock \doi{10.1016/0550-3213(91)90122-E}

\bibitem{Candelas:2016usb}
Candelas, P., de~la Ossa, X., McOrist, J.: {A metric for heterotic moduli}.
\newblock Commun. Math. Phys. \textbf{356}(2), 567--612 (2017).
\newblock \doi{10.1007/s00220-017-2978-7}

\bibitem{Carlevaro:2009jx}
Carlevaro, L., Israel, D.: {Heterotic Resolved Conifolds with Torsion, from
  Supergravity to CFT}.
\newblock JHEP \textbf{01}, 083 (2010).
\newblock \doi{10.1007/JHEP01(2010)083}

\bibitem{Chen:2013nma}
Chen, F., Dasgupta, K., Lapan, J.M., Seo, J., Tatar, R.: {Gauge/Gravity Duality
  in Heterotic String Theory}.
\newblock Phys. Rev. D \textbf{88}, 066,003 (2013).
\newblock \doi{10.1103/PhysRevD.88.066003}

\bibitem{Chisamanga:2024xbm}
Chisamanga, B., McOrist, J., Picard, S., Svanes, E.E.: {The decoupling of
  moduli about the standard embedding}  (2024)

\bibitem{Cicoli:2013rwa}
Cicoli, M., de~Alwis, S., Westphal, A.: {Heterotic Moduli Stabilisation}.
\newblock JHEP \textbf{10}, 199 (2013).
\newblock \doi{10.1007/JHEP10(2013)199}

\bibitem{Clarke:2016qtg}
Clarke, A., Garcia-Fernandez, M., Tipler, C.: {Moduli of $G_2$ structures and
  the Strominger system in dimension 7}  (2016)

\bibitem{Clarke:2020erl}
Clarke, A., Garcia-Fernandez, M., Tipler, C.: {$T$-dual solutions and
  infinitesimal moduli of the $G_2$-Strominger system}.
\newblock Adv. Theor. Math. Phys. \textbf{26}(6), 1669--1704 (2022).
\newblock \doi{10.4310/ATMP.2022.v26.n6.a3}

\bibitem{Collins:2022snx}
Collins, T.C., Picard, S., Yau, S.T.: {The Strominger system in the square of a
  K\"ahler class}  (2022)

\bibitem{Curio:2004pf}
Curio, G.: {Standard model bundles of the heterotic string}.
\newblock Int. J. Mod. Phys. A \textbf{21}, 1261--1282 (2006).
\newblock \doi{10.1142/S0217751X06025109}

\bibitem{Dasgupta:1999ss}
Dasgupta, K., Rajesh, G., Sethi, S.: {M theory, orientifolds and G - flux}.
\newblock JHEP \textbf{08}, 023 (1999).
\newblock \doi{10.1088/1126-6708/1999/08/023}

\bibitem{donaldson1985anti}
Donaldson, S.K.: Anti self-dual yang-mills connections over complex algebraic
  surfaces and stable vector bundles.
\newblock Proceedings of the London Mathematical Society \textbf{3}(1), 1--26
  (1985)

\bibitem{donaldson1998gauge}
Donaldson, S.K., Thomas, R.P.: Gauge theory in higher dimensions.
\newblock The Geometric Universe pp. 31--47 (1998)

\bibitem{Fei:2015kua}
Fei, T.: {A construction of non-K\"ahler Calabi\textendash{}Yau manifolds and
  new solutions to the Strominger system}.
\newblock Adv. Math. \textbf{302}, 529--550 (2016).
\newblock \doi{10.1016/j.aim.2016.07.023}

\bibitem{Fei:2017ctw}
Fei, T., Huang, Z., Picard, S.: {A Construction of Infinitely Many Solutions to
  the Strominger System}.
\newblock J. Diff. Geom. \textbf{117}(1), 23--39 (2021).
\newblock \doi{10.4310/jdg/1609902016}

\bibitem{Fei:2019dap}
Fei, T., Picard, S.: {Anomaly flow and T-duality}.
\newblock Pure Appl. Math. Quart. \textbf{17}(3), 1083--1112 (2021).
\newblock \doi{10.4310/PAMQ.2021.v17.n3.a11}

\bibitem{Fei:2014aca}
Fei, T., Yau, S.T.: {Invariant Solutions to the Strominger System on Complex
  Lie Groups and Their Quotients}.
\newblock Commun. Math. Phys. \textbf{338}(3), 1183--1195 (2015).
\newblock \doi{10.1007/s00220-015-2374-0}

\bibitem{Fernandez:2008wa}
Fernandez, M., Ivanov, S., Ugarte, L., Villacampa, R.: {Non-Kaehler Heterotic
  String Compactifications with non-zero fluxes and constant dilaton}.
\newblock Commun. Math. Phys. \textbf{288}, 677--697 (2009).
\newblock \doi{10.1007/s00220-008-0714-z}

\bibitem{Fu:2009zee}
Fu, J.X., Tseng, L.S., Yau, S.T.: {Local Heterotic Torsional Models}.
\newblock Commun. Math. Phys. \textbf{289}, 1151--1169 (2009).
\newblock \doi{10.1007/s00220-009-0789-1}

\bibitem{Fu:2006vj}
Fu, J.X., Yau, S.T.: {The Theory of superstring with flux on non-Kahler
  manifolds and the complex Monge-Ampere equation}.
\newblock J. Diff. Geom. \textbf{78}(3), 369--428 (2008)

\bibitem{Garcia-Fernandez:2016azr}
Garcia-Fernandez, M.: {Lectures on the Strominger system}  (2016)

\bibitem{Garcia-Fernandez:2018qcl}
Garcia-Fernandez, M.: {T-dual solutions of the Hull\textendash{}Strominger
  system on non-K\"ahler threefolds}.
\newblock J. Reine Angew. Math. \textbf{2020}(766), 137--150 (2020).
\newblock \doi{10.1515/crelle-2019-0013}

\bibitem{GarciaFernandez2020}
Garcia-Fernandez, M.: T-dual solutions of the {H}ull–{S}trominger system on
  non-{K}ahler threefolds.
\newblock Journal fur die reine und angewandte Mathematik \textbf{2020}(766),
  137--150 (2020).
\newblock \doi{doi:10.1515/crelle-2019-0013}

\bibitem{Garcia-Fernandez:2023nil}
Garcia-Fernandez, M., Molina, R.G.: {Futaki Invariants and Yau's Conjecture on
  the Hull-Strominger system}  (2023)

\bibitem{Garcia-Fernandez:2023vah}
Garcia-Fernandez, M., Molina, R.G.: {Harmonic metrics for the
  Hull\textendash{}Strominger system and stability}.
\newblock Int. J. Math. \textbf{35}(09), 2441,008 (2024).
\newblock \doi{10.1142/S0129167X24410088}

\bibitem{Garcia-Fernandez:2024ypl}
Garcia-Fernandez, M., Molina, R.G., Streets, J.: {Pluriclosed flow and the
  Hull-Strominger system}  (2024)

\bibitem{Garcia-Fernandez:2018emx}
Garcia-Fernandez, M., Rubio, R., Shahbazi, C., Tipler, C.: {Canonical metrics
  on holomorphic Courant algebroids}.
\newblock Proc. Lond. Math. Soc. \textbf{125}(3), 700--758 (2022).
\newblock \doi{10.1112/plms.12468}

\bibitem{Garcia-Fernandez:2015hja}
Garcia-Fernandez, M., Rubio, R., Tipler, C.: {Infinitesimal moduli for the
  Strominger system and Killing spinors in generalized geometry}.
\newblock Math. Ann. \textbf{369}(1-2), 539--595 (2017).
\newblock \doi{10.1007/s00208-016-1463-5}

\bibitem{Garcia-Fernandez:2018ypt}
Garcia-Fernandez, M., Rubio, R., Tipler, C.: {Holomorphic string algebroids}.
\newblock Trans. Am. Math. Soc. \textbf{373}(10), 7347--7382 (2020).
\newblock \doi{10.1090/tran/8149}

\bibitem{Garcia-Fernandez:2020awc}
Garcia-Fernandez, M., Rubio, R., Tipler, C.: {Gauge theory for string
  algebroids}.
\newblock J. Diff. Geom. \textbf{128}(1), 77--152 (2024).
\newblock \doi{10.4310/jdg/1721075260}

\bibitem{Goldstein_2004}
Goldstein, E., Prokushkin, S.: {Geometric model for complex non-Kahler
  manifolds with SU(3) structure}.
\newblock Communications in Mathematical Physics \textbf{251}(1), 65--78
  (2004).
\newblock \doi{10.1007/s00220-004-1167-7}

\bibitem{Gurrieri:2004dt}
Gurrieri, S., Lukas, A., Micu, A.: {Heterotic on half-flat}.
\newblock Phys.Rev. \textbf{D70}, 126,009 (2004).
\newblock \doi{10.1103/PhysRevD.70.126009}

\bibitem{Halmagyi:2017lqm}
Halmagyi, N., Israel, D., Sarkis, M., Svanes, E.E.: {Heterotic Hyper-K\"ahler
  flux backgrounds}.
\newblock JHEP \textbf{08}, 138 (2017).
\newblock \doi{10.1007/JHEP08(2017)138}

\bibitem{Halmagyi:2016pqu}
Halmagyi, N., Israel, D., Svanes, E.E.: {The Abelian Heterotic Conifold}.
\newblock JHEP \textbf{07}, 029 (2016).
\newblock \doi{10.1007/JHEP07(2016)029}

\bibitem{Hull:1986kz}
Hull, C.: {Compactifications of the Heterotic Superstring}.
\newblock Phys.Lett. \textbf{B178}, 357 (1986).
\newblock \doi{10.1016/0370-2693(86)91393-6}

\bibitem{deIaOssa:2019cci}
de~Ia~Ossa, X., Larfors, M., Magill, M., Svanes, E.E.: {Superpotential of three
  dimensional $ \mathcal{N} $ = 1 heterotic supergravity}.
\newblock JHEP \textbf{01}, 195 (2020).
\newblock \doi{10.1007/JHEP01(2020)195}

\bibitem{Ibarra:2024hfm}
Ibarra, J.J.M., Oehlmann, P.K., Ruehle, F., Svanes, E.E.: {A Heterotic K\"ahler
  Gravity and the Distance Conjecture}  (2024)

\bibitem{Ivanov:2009rh}
Ivanov, S.: {Heterotic supersymmetry, anomaly cancellation and equations of
  motion}.
\newblock Phys.Lett. \textbf{B685}, 190--196 (2010).
\newblock \doi{10.1016/j.physletb.2010.01.050}

\bibitem{Kupka:2024rvl}
Kupka, J., Strickland-Constable, C., Svanes, E.E., Tennyson, D., Valach, F.:
  {BPS complexes and Chern--Simons theories from $G$-structures in gauge theory
  and gravity}  (2024)

\bibitem{Kupka:2024tic}
Kupka, J., Strickland-Constable, C., Valach, F.: {Supergravity without gravity
  and its BV formulation}  (2024)

\bibitem{deLazari:2024zkg}
de~L\'azari, H., Lotay, J.D., Earp, H.S., Svanes, E.E.: {Local descriptions of
  the heterotic SU(3) moduli space}  (2024)

\bibitem{LopesCardoso:2003dvb}
Lopes~Cardoso, G., Curio, G., Dall'Agata, G., Lust, D.: {BPS action and
  superpotential for heterotic string compactifications with fluxes}.
\newblock JHEP \textbf{10}, 004 (2003).
\newblock \doi{10.1088/1126-6708/2003/10/004}

\bibitem{LopesCardoso:2002vpf}
Lopes~Cardoso, G., Curio, G., Dall'Agata, G., Lust, D., Manousselis, P.,
  Zoupanos, G.: {NonKahler string backgrounds and their five torsion classes}.
\newblock Nucl. Phys. B \textbf{652}, 5--34 (2003).
\newblock \doi{10.1016/S0550-3213(03)00049-X}

\bibitem{Martelli:2010jx}
Martelli, D., Sparks, J.: {Non-Kahler heterotic rotations}.
\newblock Adv.Theor.Math.Phys. \textbf{15}, 131--174 (2011).
\newblock \doi{10.4310/ATMP.2011.v15.n1.a4}

\bibitem{McOrist:2016cfl}
McOrist, J.: {On the Effective Field Theory of Heterotic Vacua}.
\newblock Lett. Math. Phys. \textbf{108}(4), 1031--1081 (2018).
\newblock \doi{10.1007/s11005-017-1025-0}

\bibitem{MPS24}
McOrist, J., Picard, S., Svanes, E.: A heterotic {H}ermitian--{Y}ang--{M}ills
  equivalence  (2024)

\bibitem{McOrist:2019mxh}
McOrist, J., Sisca, R.: {Small gauge transformations and universal geometry in
  heterotic theories}.
\newblock SIGMA \textbf{16}, 126 (2020).
\newblock \doi{10.3842/SIGMA.2020.126}

\bibitem{McOrist:2024ivz}
McOrist, J., Sticka, M., Svanes, E.E.: {The physical moduli of heterotic G\_2
  string compactifications}  (2024)

\bibitem{McOrist:2021dnd}
McOrist, J., Svanes, E.E.: {Heterotic quantum cohomology}.
\newblock JHEP \textbf{11}, 096 (2022).
\newblock \doi{10.1007/JHEP11(2022)096}

\bibitem{Melnikov:2014ywa}
Melnikov, I.V., Minasian, R., Sethi, S.: {Heterotic fluxes and supersymmetry}.
\newblock JHEP \textbf{06}, 174 (2014).
\newblock \doi{10.1007/JHEP06(2014)174}

\bibitem{Moroianu:2021kit}
Moroianu, A., Murcia, A., Shahbazi, C.S.: {Heterotic solitons on
  four-manifolds.}
\newblock New York J. Math. \textbf{28}, 1463--1497 (2022)

\bibitem{Moroianu:2023jof}
Moroianu, A., Murcia, A.J., Shahbazi, C.S.: {The Heterotic-Ricci Flow and Its
  Three-Dimensional Solitons}.
\newblock J. Geom. Anal. \textbf{34}(5), 122 (2024).
\newblock \doi{10.1007/s12220-024-01570-4}

\bibitem{delaOssa:2015maa}
de~la Ossa, X., Hardy, E., Svanes, E.E.: {The Heterotic Superpotential and
  Moduli}.
\newblock JHEP \textbf{01}, 049 (2016).
\newblock \doi{10.1007/JHEP01(2016)049}

\bibitem{delaOssa:2017pqy}
de~la Ossa, X., Larfors, M., Svanes, E.E.: {The Infinitesimal Moduli Space of
  Heterotic G$_{2}$ Systems}.
\newblock Commun. Math. Phys. \textbf{360}(2), 727--775 (2018).
\newblock \doi{10.1007/s00220-017-3013-8}

\bibitem{delaOssa:2014msa}
de~la Ossa, X., Svanes, E.E.: {Connections, Field Redefinitions and Heterotic
  Supergravity}.
\newblock JHEP \textbf{12}, 008 (2014).
\newblock \doi{10.1007/JHEP12(2014)008}

\bibitem{delaOssa:2014cia}
de~la Ossa, X., Svanes, E.E.: {Holomorphic Bundles and the Moduli Space of N=1
  Supersymmetric Heterotic Compactifications}.
\newblock JHEP \textbf{10}, 123 (2014).
\newblock \doi{10.1007/JHEP10(2014)123}

\bibitem{otal2017invariant}
Otal, A., Ugarte, L., Villacampa, R.: Invariant solutions to the strominger
  system and the heterotic equations of motion.
\newblock Nuclear Physics B \textbf{920}, 442--474 (2017)

\bibitem{Phong:2016gbw}
Phong, D.H., Picard, S., Zhang, X.: The anomaly flow and the fu-yau equation.
\newblock Annals of PDE \textbf{4}(2), 13 (2018)

\bibitem{Phong:2016ggz}
Phong, D.H., Picard, S., Zhang, X.: {Anomaly flows}.
\newblock Commun. Anal. Geom. \textbf{26}(4), 955--1008 (2018).
\newblock \doi{10.4310/CAG.2018.v26.n4.a9}

\bibitem{Picard:2024tqg}
Picard, S.: {The Strominger System and Flows by the Ricci Tensor}  (2024)

\bibitem{Picard:2024pad}
Picard, S., Wu, P.L.: {Balanced and Aeppli Parameters for the Heterotic Moduli}
   (2024).
\newblock \doi{10.1142/S0129167X24420023}

\bibitem{pujia2021hull}
Pujia, M.: The hull-strominger system and the anomaly flow on a class of
  solvmanifolds.
\newblock Journal of Geometry and Physics \textbf{170}, 104,352 (2021)

\bibitem{Silva:2024fvl}
Silva, A.A.d., Garcia-Fernandez, M., Lotay, J.D., Earp, H.N.S.: {Coupled
  $\operatorname{G}_2$-instantons}  (2024)

\bibitem{Smith:2024ejf}
Smith, G.R., Tennyson, D., Waldram, D.: {All-orders moduli for type II flux
  backgrounds}  (2024)

\bibitem{Smith:2022baw}
Smith, G.R., Waldram, D.: {M-theory moduli from exceptional complex
  structures}.
\newblock JHEP \textbf{08}, 022 (2023).
\newblock \doi{10.1007/JHEP08(2023)022}

\bibitem{Strominger:1986uh}
Strominger, A.: {Superstrings with Torsion}.
\newblock Nucl. Phys. \textbf{B274}, 253 (1986).
\newblock \doi{10.1016/0550-3213(86)90286-5}

\bibitem{Tellez-Dominguez:2023wwr}
Tellez-Dominguez, R.: {Chern correspondence for higher principal bundles}
  (2023)

\bibitem{thomas1997gauge}
Thomas, R.P.: Gauge theory on calabi-yau manifolds.
\newblock Ph.D. thesis, University of Oxford (1997)

\bibitem{uhlenbeck1986existence}
Uhlenbeck, K., Yau, S.T.: On the existence of hermitian-yang-mills connections
  in stable vector bundles.
\newblock Communications on Pure and Applied Mathematics \textbf{39}(S1),
  S257--S293 (1986)

\bibitem{Witten:1986kg}
Witten, L., Witten, E.: {Large Radius Expansion of Superstring
  Compactifications}.
\newblock Nucl. Phys. \textbf{B281}, 109--126 (1987).
\newblock \doi{10.1016/0550-3213(87)90249-5}

\end{thebibliography}
\end{document}